\documentclass[11pt,twoside]{article}
\usepackage[utf8]{inputenc}
\usepackage{jeffe}
\usepackage[hmargin=1in,vmargin=1in]{geometry}

\usepackage{pgfplots}
\pgfplotsset{compat=1.14}

\usepackage{cite}
\usepackage{cleveref}
\usepackage{xspace}
\usepackage{enumitem}
\usepackage{amsopn}
\usepackage{verbatim}

\usepackage[mathlines]{lineno}

\DeclareMathOperator{\Lg}{lg}
\DeclareMathOperator{\Argmin}{argmin}
\DeclareMathOperator{\Argmax}{argmax}

\DeclareMathOperator{\Polylog}{\polylog}

\DeclareMathOperator{\dist}{dist}
\DeclareMathOperator{\probability}{\mathbb{P}}

\def\distG{\dist_{G^*}}
\def\eps{\varepsilon}
\def\cost{\textsc{Cost}}
\def\spread{\textsc{Sp}}
\def\dartof#1{\vec{#1}}
\def\dartsof#1{\dartof{#1}}
\def\event{e}

\newtheorem{lemma}{Lemma}[section]
\newtheorem{theorem}[lemma]{Theorem}

\begin{document}







\begin{titlepage}
\title{A near-linear time approximation scheme for geometric transportation with arbitrary supplies and spread}

\author{
Kyle Fox%
\thanks{Department of Computer Science,
    The University of Texas at Dallas; \url{kyle.fox@utdallas.edu}.}
\and Jiashuai Lu%
\thanks{Department of Computer Science,
    The University of Texas at Dallas; \url{jiashuai.lu@utdallas.edu}.}
}


\maketitle

\begin{abstract}
The geometric transportation problem takes as input a set of points \(P\) in \(d\)-dimensional Euclidean space and a supply function \(\mu : P \to \R\).
The goal is to find a transportation map, a non-negative assignment \(\tau : P \times P \to \R_{\geq 0}\) to pairs of points, so the total assignment leaving each point is equal to its supply, i.e., \(\sum_{r \in P} \tau(q, r) - \sum_{p \in P} \tau(p, q) = \mu(q)\) for all points \(q \in P\).
The goal is to minimize the weighted sum of Euclidean distances for the pairs, \(\sum_{(p, q) \in P \times P} \tau(p, q) \cdot ||q - p||_2\).

We describe the first algorithm for this problem that returns, with high probability, a \((1 + \eps)\)-approximation to the optimal transportation map in \(n\eps^{-O(d)}\log^{O(d)}{n}\) time.
In contrast to the previous best algorithms for this problem, our near-linear running time bound is independent of the spread of \(P\) and the magnitude of its real-valued supplies.
\end{abstract}

\setcounter{page}{0}
\thispagestyle{empty}
\end{titlepage}

\pagestyle{myheadings}
\markboth{Kyle Fox and Jiashuai Lu}
		{\footnotesize A near-linear time approximation scheme for geometric transportation with arbitrary supplies and spread \normalsize}


\section{Introduction}\label{sec:intro}
We consider the \EMPH{geometric transportation problem} in \(d\)-dimensional Euclidean space for any constant \(d\).
In this problem, we are given a set \(P \subset \R^d\) of \(n\) points.
Each point is assigned a real \EMPH{supply} \(\mu : P \to \R\) where \(\sum_{p \in P} \mu(p) = 0\).
A \EMPH{transportation map} is a non-negative assignment \(\tau : P \times P \to \R_{\geq 0}\) to pairs of points such that for all \(q \in P\) we have \(\sum_{r \in P} \tau(q, r) - \sum_{p \in P} \tau(p, q) = \mu(q)\).
The \EMPH{cost} of the transportation map is the weighted sum of Euclidean distances across all pairs, i.e. \(\sum_{(p, q) \in P \times P} \tau(p, q) \cdot ||q - p||_2\).
Our goal is to find a transportation map of minimum cost, and we denote this minimum cost as \(\cost(P, \mu)\).

One may imagine the points with positive supply as piles of earth and those with negative supplies as holes in the ground.
A transportation map describes how to transfer the earth to the holes without overfilling any hole, and its cost is the total number of ``earth-miles'' used to do the transfer.
Consequently, \(\cost(P, \mu)\) is often referred to as the earth mover's distance, although it can also be called the \(1\)-Wasserstein distance between measures over the positively and negatively supplied points.
The continuous version of the problem is sometimes called the \emph{optimal transport} or \emph{Monge-Kantorovich} problem, and it has been studied extensively by various mathematics communities~\cite{villani2008optimal}.
The discrete version we study here has applications in shape matching, image retrieval, and graphics~\cite{DBLP:conf/cvpr/GraumanD04,DBLP:journals/ijcv/RubnerTG00,DBLP:journals/tog/BonneelPPH11,DBLP:journals/tog/SolomonRGB14,DBLP:conf/icml/CuturiD14,DBLP:conf/eccv/GiannopoulosV02}.

Computing a transportation map that minimizes weighted Euclidean distances (or any other metric) is easily done in polynomial time by reduction to the uncapacitated minimum cost flow problem in a complete bipartite graph between points with positive supply and those with negative supply.
The graph has as many as \(\Omega(n^2)\) edges, so this approach takes~\(O(n^3 \Polylog{n})\) time using a combinatorial minimum cost flow algorithm of Orlin~\cite{DBLP:journals/ior/Orlin93}.
Assuming integral supplies with absolute values summing to \(U\), we can use an algorithm of Lee and Sidford~\cite{DBLP:conf/focs/LeeS14} instead to reduce the running time to \(O(n^{2.5} \Polylog{(n, U)})\).
Taking advantage of the geometry inherent in the problem, Agarwal \emph{et al.}~\cite{DBLP:conf/compgeom/AgarwalFPVX17} describe how to implement Orlin's algorithm for arbitrary supplies to find the optimal transportation map in~\(O(n^2 \Polylog{n})\) time, but only for \(d = 2\).

We can significantly reduce these running times by accepting a small loss in optimality.
However, many results along this line focus on merely estimating the earth mover's distance instead
of computing the associated transportation map.
Indyk~\cite{DBLP:conf/soda/Indyk07} describes an \(O(n \Polylog{n})\) time algorithm that estimates the earth mover's distance within a constant factor assuming unit supplies.
Andoni \emph{et al.}~\cite{DBLP:conf/stoc/AndoniNOY14} describe an \(O(n^{1+ o(1)})\) time algorithm for arbitrary supplies that estimates the cost within a \(1 + \eps\) factor (the dependency on \(\eps\) is hiding in the \(o(1)\)).
As pointed out by Khesin, Nikolov, and Paramonov~\cite{khesin2021preconditioning}, a \(1 + \eps\) factor estimation of the distance is possible in \(n^{1 + o(1)} \eps^{-O(d)}\) time (without the \(o(1)\) hiding dependencies on \(\eps\)) by running an approximation algorithm for minimum cost flow by Sherman~\cite{DBLP:conf/soda/Sherman17} on a sparse Euclidean spanner over the input points.
However, it is not clear how to extract a nearly optimal transportation map using the spanner's flow.

Finding an actual transportation \emph{map} may be more difficult.
Sharathkumar and Agarwal~\cite{DBLP:conf/soda/SharathkumarA12} describe a \((1 + \eps)\)-approximation algorithm for the integral supply case (i.e., an algorithm returning a map of cost at most \((1 + \eps) \cdot \cost(P, \mu)\)) in \(O(n \sqrt{U} \Polylog{(U, \eps, n)})\) time.
Agarwal \emph{et al.}~\cite{DBLP:conf/compgeom/AgarwalFPVX17} describe a randomized algorithm with expected \(O(\log^2 (1 / \eps))\)-approximation ratio running in \(O(n^{1 + \eps})\) expected time for the arbitrary supply case and a deterministic \(O(n^{3/2} \eps^{-d} \Polylog{(U, n)})\) time \((1+\eps)\)-approximation algorithm for the bounded integral supply case.
Lahn \emph{et al.}~\cite{DBLP:conf/nips/LahnMR19} describe an \(O(n(C\delta)^2\Polylog{(U, n)})\)
(\(C=\max_{p\in P}{|\mu(p)|}\)) time algorithm computing a map of cost at most \(\cost(P,\mu)+\delta U\).
Recently, Khesin \emph{et al.}~\cite{khesin2021preconditioning} described a randomized \((1 + \eps)\)-approximation algorithm for the arbitrary supply case running in 
\(n \eps^{-O(d)} \log^{O(d)}{\spread(P)} \log{n}\) time, where \(\spread(P)\) is the spread of the point set~\(P\).%
\footnote{Khesin \emph{et al.}~\cite{khesin2021preconditioning} and Agarwal \etal~\cite{DBLP:conf/compgeom/AgarwalFPVX17} present geometric transportation with integer supplies, but their unbounded supply algorithms work without modification when presented with real valued supplies.}
The \emph{spread} (also called aspect ratio) of \(P\) is the ratio of the diameter of \(P\) to the smallest pairwise distance between points in \(P\).
As Khesin \emph{et al.} point out, one can reduce an instance with unbounded spread but bounded integral supplies to the case of bounded spread to get a \((1 + \eps)\)-approximation running in \(n \eps^{-O(d)} \log^{O(d)}{U} \log^2{n}\) time, generalizing a near-linear time \((1 + \eps)\)-approximation algorithm by Sharathkumar and Agarwal~\cite{raghvendra2020near} for the unit supply case.
The unit supply case is sometimes referred to as the \emph{geometric bipartite matching problem}. Agarwal and Sharathkumar~\cite{DBLP:conf/stoc/AgarwalS14} also describe a \emph{deterministic} \((1 / \eps)\)-approximation algorithm for geometric bipartite matching that runs in \(O(n^{1+\eps}\log{n})\) time.

Despite these successes, prior work still does not include a near-linear time \((1 + \eps)\)-approximation algorithm for the general case of arbitrary spread and real valued supplies.
Often, an algorithm designed for bounded spread cases can be extended to work with cases of arbitrary spread.
For example, one might substitute in \emph{compressed quadtrees}~\cite[Chapter 2]{har2011geometric} in places where the bounded spread algorithm uses standard (uncompressed) quadtrees.
This straightforward approach does not appear to work for the geometric transportation problem, however.
As detailed below, Khesin \emph{et al.}~\cite{khesin2021preconditioning} use a quadtree to build a sparse graph as part of a reduction to the minimum cost flow problem.
Both their running time and approximation analysis rely heavily on the tree having low depth when the spread is bounded.
Unfortunately, a compressed quadtree is only guaranteed to have small \emph{size};
the depth can still be linear in the number of leaves.
One may also try the strategy of separating out groups of points~\(P'\) that are much closer to each other than to the rest of the point set~\(P\), routing as much supply as possible within~\(P'\), and then satisfying what remains of the supplies in~\(P'\) by treating~\(P'\) as a single point.
In fact, the result described below does employ a variant of this strategy (see Section~\ref{sec:spanner-decomposition}).
However, the simplified instances of the problem one gets using this strategy still yield compressed quadtrees of very high depth.

\subsection{Our results and approach}

We describe a randomized \((1 + \eps)\)-approximation algorithm for the geometric transportation problem that runs in near-linear time irrespective of the spread of \(P\) or the supplies of its points.
Our specific result is spelled out in the following theorem.
We say an event occurs with \emph{high probability} if it occurs with probability at least \(1 - 1 / n^c\) for some constant \(c\).
\begin{theorem}\label{theorem-result}
There exists a randomized algorithm that, given a set of \(n\) points \(P \in \R^d\) and a supply function \(\mu : P \to \R\), runs in time \(n \eps^{-O(d)} \log^{O(d)} n\) and with high probability returns a transportation map with cost at most \((1 + \eps) \cdot \cost(P, \mu)\).
\end{theorem}

At a high level, our algorithm follows the approach laid out by Khesin \emph{et al.}~\cite{khesin2021preconditioning} for the bounded spread case.
However, removing the running time's dependency on the spread introduces fundamental and technical issues to nearly every step in their approach.

Let \(\eps_0\) be a function of \(\eps\) and \(n\) to be specified later.
Taking a cue from prior work on geometric transportation and its specializations~\cite{raghvendra2020near,DBLP:conf/stoc/AndoniNOY14},
Khesin \emph{et al.}'s algorithm begins by building a random sparse graph over \(n \eps_0^{-O(d)} \log \spread(P)\) vertices including the points in \(P\).
In expectation, the shortest path distance between any pair of points in \(P\) is maintained up to an \(O(\eps_0 \log \spread(P))\) factor, so computing a transportation map is done by setting \(\eps_0\) to \(O(\eps / (\log \spread(P))\) and running a minimum cost flow algorithm on the sparse graph.

The graph is constructed by first building a randomly shifted quadtree over \(P\).
The quadtree is constructed by surrounding \(P\) with an axis-aligned box called a cell, partitioning it into \(2^d\) equal sized child cells, and recursively building a quadtree in each child cell;
the whole tree has depth \(\log \spread(P)\).
After building the quadtree, they add \(\eps_0^d\) Steiner vertices within each cell along with a carefully selected set of edges.
While other methods are known for constructing such a sparse graph even without Steiner vertices~\cite{DBLP:conf/soda/CallahanK93} and we can directly approximate the transporation cost on a sparse Euclidean spanner in near linear time through Sherman's framework~\cite{DBLP:conf/soda/Sherman17}, the hierarchical structure of Khesin \emph{et al.}'s construction is necessary for extracting the transportation map after a minimum cost flow is computed.
Observe that not only is the quadtree's size dependent on \(\spread(P)\), but so is the number of Steiner vertices added to each cell.

As suggested earlier, the natural approach for reducing the quadtree's size is to remove subtrees containing no members of \(P\) and to \emph{compress} the tree by replacing each maximal path of cells with exactly one non-empty child each with a single link to the lowest cell in the path.
This approach does result in a quadtree of size \(O(n)\), but its depth could also be as large as \(\Omega(n)\).
This large depth introduces many issues, the worst of which is that standard analyses only guarantee
point-to-point distances are maintained up to an \(O(\eps_0 n)\) factor.
We cannot afford to set \(\eps_0\) to \(\eps / n\), because the sparse graph would have
\(\Omega(n^d)\) vertices!

The solution to avoiding such a large increase in expected distances is to use the idea of \emph{moats} around the points as done in the almost-linear time constant factor approximation algorithm of Agarwal \emph{et al.}~\cite{DBLP:conf/compgeom/AgarwalFPVX17}.
In short, we modify the quadtree construction so that, with high probability, all points are sufficiently far away from the boundary of every quadtree cell they appear in.
Assuming this condition holds, there are only a limited number of quadtree ``levels'' at which a pair of points can be separated, and we use this fact to show distances increase by only an \(O(\eps_0 \log n)\) factor in expectation.
It turns out modifying the quadtree construction correctly is a surprisingly subtle task.
Guaranteeing the moats are avoided potentially requires us to perform independent random shifts at several places throughout the quadtree.
However, we need to be selective with where the independent shifts occur so that we can successfully analyze the expected distances between points in the sparse graph.

The second stage of Khesin \emph{et al.}'s~\cite{khesin2021preconditioning} algorithm solves the minimum cost flow problem in the sparse graph using a framework of Sherman~\cite{DBLP:conf/soda/Sherman17}.
First, they encode the minimum cost flow problem as finding a flow vector \(f\) of minimum cost subject to linear constraints \(A f = b\) where \(A\) is the vertex-edge incidence matrix and \(b\) is a supply vector (not necessarily equal to \(\mu\)).
Sherman's framework involves repeatedly finding flows \(f\) of approximately optimal cost that approximately satisfy such constraints.
Each iteration of this algorithm requires an application of \(A\) and \(A^T\) to a pair of vectors, and
the number of iterations needed in this approach is polynomial in the \emph{condition number} of \(A\).
Unfortunately, \(A\) may not be well-conditioned, so Khesin \emph{et al.} describe a \emph{preconditioner} matrix \(B\) such that \(BA\) has low condition number and is still sparse.
They proceed to use Sherman's framework under the equivalent constraints \(BA f = Bb\).

One interpretation of Khesin \emph{et al.}'s~\cite{khesin2021preconditioning} preconditioner is that it describes a way to charge each Steiner vertex an amount based on the supply of ``descendent'' vertices below it so that the sum of charges bound the cost of an optimal flow from below.
Consequently, both the number of non-zero entries in each column of \(B\) and the condition number of \(B\) are proportional to the quadtree's depth.

The high depth of our quadtree again appears to cause issues.
However, our use of moats implies additional structure to the sparse graph that we can take
advantage of.
Our preconditioner \(B\) is based on essentially the same charging scheme as Khesin \emph{et al.}, but thanks to the moats, we prove the condition number remains proportional to \(O(\eps_0^{-1} \log (n / \eps_0))\) instead of the quadtree depth.
This charging scheme still results in a preconditioner \(B\) that \emph{is not} sparse, so a naive implementation of Sherman's~\cite{DBLP:conf/soda/Sherman17} framework may take quadratic time per iteration.
To address this issue, we describe a pair of algorithms based on the hierarchical structure of the graph that let us apply both \(BA\) and its transpose in only linear time.

The final stage of the algorithm is the extraction of an approximately minimum cost transportation map from an approximately minimum cost flow in the sparse graph.
Khesin \emph{et al.}'s~\cite{khesin2021preconditioning}'s original procedure modifies the graph's flow by iteratively reassigning flow to travel directly from input points to each of their many ancestor Steiner vertices or vice versa.
We use binary search tree based data structures in a novel way to do flow reassignments in bulk, allowing us to extract the transportation map in time near-linear in the graph size.

Our result relies on a computation model supporting basic combinatoric operations over real numbers
including addition, subtraction, multiplication, division, and comparisons.
In particular, we do not rely on the floor function.%
\footnote{A previous version~\cite{fl-ntasg-20} of this paper relied on a slightly stronger model of
computation that allowed one to quickly compute the location of points within arbitrary grids (see
Bern \etal~\cite{bet-pcqqt-99} and Har-Peled~\cite[Chapter 2]{har2011geometric}).
The current version uses a slightly modified algorithm and more careful analysis to avoid needing
this operation.}
Our results (and those of Khesin \emph{et al.}~\cite{khesin2021preconditioning}) can be extended to work with any \(L_p\) metric instead of just Euclidean distance.
The rest of the paper proceeds as follows.
We describe our sparse graph construction and describe the reduction to minimum cost flow in Section~\ref{sec:spanner}.
We describe our preconditioner and its use Section~\ref{sec:sol}.
Finally, we describe how to extract the approximately optimal transportation map from a flow on the sparse graph in Section~\ref{sec:recover}.

\section{Reduction to minimum cost flow in a sparse graph}\label{sec:spanner}
In this section, we present a way to build a sparse graph~\(G^* = (V^*, E^*)\) based on \(P\) and reduce the transportation problem to finding a minimum cost flow in this sparse graph.
Similar to the one presented by Khesin \etal~\cite{khesin2021preconditioning}, our sparse graph \(G^*\) is based on a randomly shifted quadtree whose cells have been subdivided into smaller \emph{subcells}.
However, the quadtree we use is compressed under certain conditions to guarantee the number of nodes in it is nearly linear in \(n\).
We also independently shift subtrees under these compresion points so we may guarantee a low
expected distortion for point-to-point distance.

\subsection{Construction of the sparse graph}


Our sparse graph construction begins by building a variant of the compressed quadtree on \(P\) we call a \EMPH{conditionally-compressed quadtree}.
Let \(T^*\) denote this tree.
Let \(\Box_P\) be the minimum bounding square of \(P\).
We fix an \(\eps_0 = O(\eps / \log n)\) such that \(1 / \eps_0\) is a power of \(2\). 
Suppose the side length of \(\square_P\) is \(\Delta^*\).
Let \({\square}\) be a square of side length \(3\Delta^*\) such that \(\square_P\) and \({\square}\) are concentric.
We shift \({\square}\) by a vector chosen uniformly at random from \([0,\Delta^*)^d\). See Figure~\ref{fig:quadtree}, left.

\begin{figure}
\begin{minipage}{0.4\textwidth}
\begin{center}
    \includegraphics[width=0.8\textwidth]{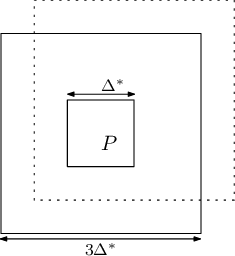}
\end{center}
\end{minipage}\hfill
\begin{minipage}{0.5\textwidth}
\begin{center}
 \includegraphics[width=0.8\textwidth]{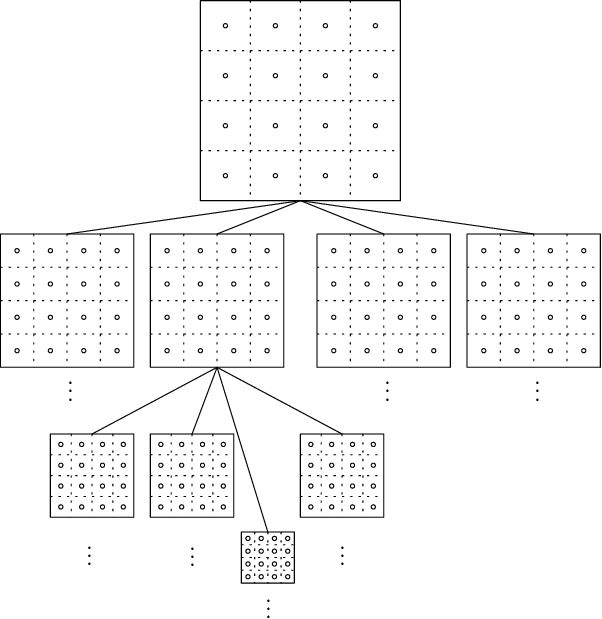}
\end{center}
\end{minipage}
\caption{Left: Randomly shifting a box around \(P\). Right: The quadtree cells form a hierarchy. Each cell is partitioned into \(\eps_0^{-d}\) sub cells, and each subcell has a single net point at its center.}
\label{fig:quadtree}
\end{figure}

Each node of \(T^*\) is a square cell in \(\mathbb{R}^d\).
Set \({\square}\) to be the root of \(T^*\).
We recursively process each cell \(C\) as follows. Suppose \(C\) has side length \(\Delta\) and the subset of \(P\) in \(C\) is \(P'\).
Let \(\Delta_{P'}\) be the side length of the minimum bounding square \(\square_{P'}\) of \(P'\).

\begin{enumerate}[label=\arabic*)]
\item If \(|P'|=1\), then~\(C\) is a leaf node.
\item If \(|P'|>1\) and \(\Delta_{P'}<\frac{\eps_0\Delta}{3n^8}\), we find the minimum bounding square \(\square_{P'}\) of \(P'\).
  Let \(\Delta_{P'}\) be the side length of \(\square_{P'}\).
  We recursively build a conditionally-compressed quadtree over \(P'\) with an independently shifted root square \(\square'\) with side length \(3\Delta_{P'}\) that is concentric to \(\square_{P'}\) before the shift.
  We connect the root of this sub-quadtree to \(T^*\) as a child of \(C\).
\item If \(|P'|>1\) and \(\Delta_{P'}\ge \frac{\eps_0\Delta}{3n^8}\), we evenly divide \(C\) into \(2^d\) squares in \(\mathbb{R}^d\) each of side length \(\frac{\Delta}{2}\), and make each square that contains at least one point of \(P'\) a child cell of \(C\).
\end{enumerate}

\begin{lemma}
\label{lm:construction}
Let \(m\) be an upper bound on the number of nodes in~\(T^*\).
Conditionally-compressed quadtree~\(T^*\) can be constructed in \(O(m + n \log n)\) time.
\end{lemma}
\begin{proof}
  Suppose we are processing a cell \(C\) containing point subset \(P'\).
  Following standard practice~\cite{DBLP:conf/soda/CallahanK93}, we assume access to \(d\) doubly-linked lists containing the points of \(P'\).
  The points in each list are sorted by distinct choices of one of their \(d\) coordinates.

  We now describe how to process \(C\).
  We determine if \(|P'|=1\) in constant time.
  If so, we stop processing \(C\) and discard its data structures.
  Otherwise, we use the lists to determine \(\Delta_{P'}\) and \(\square_{P'}\) in \(O(1)\) time.
  If \(\Delta_{P'}<\frac{\eps_0\Delta}{3n^8}\), we pass along the lists for \(C\) to the recursive quadtree constructions as Rule 2 describes.

  Suppose \(\Delta_{P'}\ge \frac{\eps_0\Delta}{3n^8}\) and Rule 3 applies.
  We compute the point subsets and their lists going into each child cell by splitting \(P'\) one
  dimension at a time.
  Specifically, for each dimension, we search the relevant linked list from both ends
  simultaneously for the position of the split.
  In doing so, we find the right position to split in time proportional to the less populated side
  of the split.
  In the same amount of time, we also perform individual deletions and insertions to make the
  other linked lists for the points on the less populated side;
  what remains of the original linked lists are the points going to the more populated side.
  Eventually, we pass along the lists we construct when computing subtrees for children of \(C\).

We spend constant time per point moving to a new list data structure when applying Rule 3. Recall the most populated child cell inherits the data structure from its parent. Therefore, whenever a point moves to a new data structure, the number of points in its cell drops by a factor of at least~\(2\). So every point gets touched at most \(\log{n}\) times and we spend \(O(m+n\log{n})\) time total implementing Rule 3.
\end{proof}

We define a type of sub-quadtree of \(T^*\).
A \EMPH{simple sub-quadtree} is a sub-quadtree consisting of a cell \(C\) that either is the root of \(T^*\) or is randomly shifted independently of its parent along with a maximal set of descendent cells of \(C\) that were \emph{not} shifted independently of \(C\) (i.e., Rule 2 was never applied to create descendent cells of \(C\) in the sub-quadtree).

For every cell \(C\) in \(T^*\), we perform a secondary subdivision on \(C\).
Let \(\Delta_C\) denote the side length of \(C\).
We divide \(C\) into \(\eps_0^{-d}\) square regions with equal side length \(\eps_0\Delta_C\).
If a sub-region of \(C\) contains a point \(p\in P\), we say it is a \EMPH{subcell} \(\Tilde{C}\) of \(C\) and we use \(C^+\) to denote the set of subcells of \(C\). Again, see Figure~\ref{fig:quadtree}.

Utilizing an idea of Agarwal~\emph{et al.}~\cite{DBLP:conf/compgeom/AgarwalFPVX17}, we define the \EMPH{moat of size \(h\)} around a point \(p\) as an axis-parallel square of side length \(h\) around \(p\).
Consider a randomly shifted grid with cells of side length \(\Delta\).
The probability of any of the grid lines hitting a moat of size \(\frac{2\Delta}{n^4}\) around any points \(p\in P\) is at most \(\frac{2\Delta}{n^4}\cdot n\cdot \frac{d}{\Delta}=O(\frac{1}{n^3})\).

\begin{lemma}\label{lm:tree}
  With probability at least \(1-O((1/n)\log(n/\eps_0))\), the conditionally-compressed quadtree \(T^*\) has the following properties:
  \begin{enumerate}[label=\arabic*.]
  \item The total number of cells is \(O(n\log(n/\eps_0))\).
  \item Suppose cell \(C\) with side length \(\Delta_C\) contains \(p\in P\) and let \(\Tilde{C}\) be the subcell of \(C\) that contains \(p\).
    Then \(p\) is at least \(\frac{\Delta_C}{n^4}\) distance away from any side of \(C\) and is at least \(\frac{\Delta_C\eps_0}{n^4}\) distance away from any side of \(\Tilde{C}\).
    In other words, the moats of \(p\) with respect to the uniform grids containing \(C\) and \(\Tilde{C}\) as cells do not touch the grid lines.
  \item Let \(T'\) be any simple sub-quadtree of \(T^*\), and let \(C'\) be a child cell of some leaf \(C\) of \(T'\).
    Cell \(C'\) lies entirely within a \emph{sub}cell of \(C\).
  \end{enumerate}
\end{lemma}
\begin{proof}
  The condition to trigger Rule 2 guarantees every path of descendent cells with one child each has length \(O(\log(n/\eps_0))\).
  We get Property 1.
  
  Let \(T_0\) be the simple sub-quadtree containing the root cell of \(T^*\).
  With Property 1, the number of cells in \(T_0\) is \(O(n\log{(n/\eps_0)})\) and hence the smallest cell in \(T_0\) lies \(O(n\log{(n/\eps_0)})\) quadtree levels down.
  Therefore, at most \(O(n\log{(n/\eps_0)})\) shifted grids in \(\mathbb{R}^d\) determine the boundaries of \(T_0\)'s (sub)cells.
  We see Property 2 is violated for at least one cell in \(T_0\) with probability at most \(c\cdot \frac{n}{n^3}\cdot \log{\frac{n}{\eps_0}}\) for some constant \(c\).
  
  Assume from here on that Property 2 holds for all cells in \(T_0\).
  Let \(C\) be any leaf cell with side length \(\Delta_C\) of \(T_0\) such that \(C\) has some descendent cell in \(T^*\).
  Let \(C'\) be a child cell with side length \(\Delta_{C'}\) of \(C\) and \(T'\) be the simple sub-quadtree rooted at \(C'\).
  Then we have \(\Delta_{C'}<\frac{\eps_0\Delta_C}{n^8}\) by Rule 2.
  Moreover, all points in \(T'\) are at distance at least \(\frac{\eps_0\Delta_C}{n^4}\) away from subcell boundaries of \(C\). Any random shift of \(T'\) using a shift vector in \([0, \frac{\Delta_{C'}}{3})^d\) will keep all points in \(T'\) inside the same subcell of \(C\).

  Let \(\{T_1,T_2,\dots\}\) denote the distinct sub-quadtrees, each consisting of a child of a leaf in \(T_0\) and all the child's desendents in \(T^*\).
  For each \(T_i\), let \(n_i\) be the number of points over which \(T_i\) is built.
  We have \(n_i\le n-1\) for each \(i\) and \(\sum_{i}n_i\le n\).
  Inductively assume Properties 2 and 3 fail to hold for \(T_i\) with probability at most
  \(c\cdot \frac{n^2_i}{n^3}\cdot \log\frac{n}{\eps_0}\le
  c\cdot\frac{(n-1)n_i}{n^3}\cdot\log\frac{n}{\eps_0}\).
  (For each \(T_i\), on the path from its root to each of its leaves in \(T^*\), at most \(n_i\)
  distinct simple sub-quadtrees exists. Each of those sub-quadtrees has at most
  \(n_i\log(n/\eps_0)\) levels.
  Summing them all, we may assume Properties 2 and 3 fail to hold for \(T_i\) with probability
  at most \(c\cdot \frac{n^2_i}{n^3}\cdot \log\frac{n}{\eps_0}\)).
  Taking a union bound, the probability of Properties 2 and 3 failing to hold for either \(T_0\) or any \(T_i\) is at most \(c\cdot \frac{n^2}{n^3}\cdot \log\frac{n}{\eps_0}=c\cdot \frac{1}{n}\cdot \log\frac{n}{\eps_0}\).
  (Note that if an independent random shift is not applied for each simple sub-quadtree in \(T^*\), for the root cell of some \(T_i\) with its side length much smaller than the side length of the root cell of \(T_0\), there are too many potential distinct grid sizes for the next cell down in a traditional compressed quadtree. Our union bound over all potential grid sizes would fail).
\end{proof}

We assume from here on that the properties described above do hold, but \(T^*\) is still randomly constructed conditional on those properties.
We now build the sparse graph \(G^*\) based on the decomposition.

For every cell \(C\), we add a \EMPH{net point} \(\nu\) at the center of every subcell of \(C\), and use \(N_{\Tilde{C}}\) to denote the net point of a subcell \(\Tilde{C}\).
We add \(O(\eps_0^{-2d})\) edges to build a clique among net points of subcells in \(C^+\).
Furthermore, if \(C\) has a parent cell \(C^p\), for each \(\Tilde{C}\in C^+\), there exists a \(\Tilde{C}^{p}\in C^{p^+}\) such that \(\Tilde{C}\) is totally contained in \(\Tilde{C}^{p}\), because \(1 / \eps_0\) is power of \(2\).
We add an edge connecting \(N_{\Tilde{C}^{p}}\) with \(N_{\Tilde{C}}\). We say \(\Tilde{C}^{p}\) is the \EMPH{parent subcell} of \(\Tilde{C}\) and \(N_{\Tilde{C}^{p}}\) is the \EMPH{parent net point of} \(N_{\Tilde{C}}\).
\EMPH{Children subcells} and \EMPH{children net points} are defined analogously.
Edges are weighted by the Euclidean distance of their endpoints.
Let \(\tilde{C}(p)\) denote the smallest subcell containing \(p\).
As a last step, for every  point \(p\in P\), we add an edge connecting \(p\) to \(N_{\tilde{C}(p)}\).

Let \(V^*\) be the union of \(P\) and the set of all net points we just added, and let \(E^*\) be the set of edges we added above. In short, \(V^*=\cup_{C\in T}\{{N_{\Tilde{C}}:\Tilde{C}\in C^+\}} \cup P\) and \(E^*=\cup_{C\in T^*}\{\{uv: u, v\in \{N_{\Tilde{C}}:\Tilde{C}\in C^+\}, u\ne v\}\cup\{N_{\Tilde{C}} N_{\Tilde{C}^p}, \Tilde{C}\in C^+\}\} \cup \Set{p N_{\tilde{C}(p)}, p \in P}\).
The sparse graph upon which we solve minimum cost flow is denoted \(G^*=(V^*, E^*)\).

\begin{lemma}\label{lm:espd}
  Conditioned on Property 2 of Lemma~\ref{lm:tree}, the expected distance between any pair \(p,q\in
  P\) in \(G^*\) is at most \(\Paren{1+O(\eps_0\log n)}||p-q||_2\).
\end{lemma}
\begin{proof}
Let \(\lambda=||p-q||_2\), and let \(\distG(p,q)\) be the distance between \(p\) and \(q\) in
\(G^*\).
Points \(p\) and \(q\) must be connected through the net points of some cell containing both of them.
Let \(C(p,q)\) be the lowest common ancestor cell of \(p\) and \(q\).
Let \(N_{C(p,q)}(p)\) and \(N_{C(p,q)}(q)\) be the net points of the subcell of \(C(p,q)\) that contains \(p\) and \(q\), respectively. Then \(\distG(p,q)=\distG(p, N_{C(p,q)}(p))+\distG( N_{C(p,q)}(p),N_{C(p,q)}(q))+\distG(q, N_{C(p,q)}(q))\).
Value \(\distG(p, N_{C(p,q)}(p))\) is the distance from \(N_{C(p,q)}(p)\) to \(p\) through its descendent net points.
We have \(\sum_{i\ge 1}2^{-i}\sqrt{d}\eps_0\Delta_{C(p,q)}\le\sqrt{d}\eps_0\Delta_{C(p,q)}\), because subcell side lengths at least halve every level down in \(T^*\).
Similarly, 
\(\distG(q, N_{C(p,q)}(q))\le\sqrt{d}\eps_0\Delta_{C(p,q)}\).
By the triangle inequality, \(\distG( N_{C(p,q)}(p),N_{C(p,q)}(q))\le||p-q||_2+||p-N_{C(p,q)}(p)||_2+||q-N_{C(p,q)}(q)||_2\le||p-q||_2+\sqrt{d}\eps_0\Delta_{C(p,q)}\).
Then we have \(\distG(p,q)\le||p-q||_2+3\sqrt{d}\eps_0\Delta_{C(p,q)}\).

We define the \textit{extra cost} to be \(\Phi_{p,q}=\distG(p, q)-||p-q||_2\).
We have \(\Phi_{p,q} \le 3\sqrt{d}\eps_0\Delta_{C(p,q)}\).
Let \(e\) be the event that Property 2 of Lemma~\ref{lm:tree} holds.
We have the conditional expectation of the extra cost \(\mathbb{E}(\Phi_{p,q} \mid e)\:\le\:
\mathbb{E}(3\sqrt{d}\eps_0\Delta_{C(p,q)} \mid e)\:\le3\sqrt{d}\eps_0\mathbb{E}(\Delta_{C(p,q)} \mid
e)\).

Let \(\Delta\) be any value in \([\frac{\lambda}{\sqrt{d}}, n^4\lambda]\).
Our goal is to provide an upper bound on the probability that \(\Delta_{C(p,q)}\in [\Delta,2\Delta)\)
conditioned on Property 2 of Lemma~\ref{lm:tree}.
As we will see later, we need not consider other possible values of \(\Delta\) given we condition on
that property.
To find our probability bound, we first bound the probability conditioned on a weaker event.
Let \(C'\) denote the \emph{highest} common ancestor cell of \(p,q\) in \(T^*\) with side length at
most \(n^8\lambda\).
Let \(\Delta'\) be its side length and let \(P'\) be the point set in \(C'\).
Let \(T\) be the simple sub-quadtree in which \(C'\) lies.
In each dimension, the shift of \(T\) in that dimension that determines \(C'\) constitutes a
continuous unique interval.
Let \(\ell_j\) denote the length of this interval for the \(j\)th dimension.
For our proof, we must consider different shifts of \(T\) that lead to our current choice of \(C'\),
and argue \(\mathbb{E}(\Delta_{C(p,q)})\) is not too large.
To do so, however, we must condition on the following event, denoted \(\event'\), that there exists a
shift of \(T\) such that no point of \(P'\) lies at distance less than \(\frac{\Delta'}{n^4}\) from
the sides of \(C'\).
Note that \(\event'\) is necessary \emph{but not sufficient} for Property 2 of Lemma~\ref{lm:tree} to
hold.
We now consider the probability of \(\Delta_{C(p,q)}\in [\Delta,2\Delta)\) conditioned on \(\event'\),
denoted by \(\probability\Brack{\Delta_{C(p,q)}\in [\Delta,2\Delta) \mid \event'}\).

Let \(\Delta''=2^{-i}\cdot \Delta'\) for some \(i\in \mathbb{N}\) such that \(\Delta''\in [\Delta, 2\Delta)\).
Conditioned on a particular choice of \(C'\), each \((d-1)\)-dimensional hyperplane perpendicular to
the \(j\)-th axis in the grid of size \(\frac{\Delta''}{2}\) aligned with the shift of \(T\) has a
\(j\)-coordinate somewhere on an interval of length \(\ell_j\).
Therefore, there are at most \(\ceil{\frac{2\ell_j}{\Delta''}}\) \((d-1)\)-dimensional hyperplanes
possibly shifted to lie between \(p\) and \(q\).
Summing over all \(d\) choices of \(j\), the probability that \(p\) and \(q\) are separated by some
\((d-1)\)-dimensional hyperplane in the grid of size \(\frac{\Delta''}{2}\) is at most
\(\sum^d_{j=1}{\ceil{\frac{2\ell_j}{\Delta''}}\cdot \lambda\cdot \frac{1}{\ell_j}}\) conditioned on
\(\event'\).
Therefore, \(\probability\Brack{\Delta_{C(p,q)}\in [\Delta,2\Delta) \mid \event'} \leq
\sum^d_{j=1}{\ceil{\frac{2\ell_j}{\Delta''}}\cdot \lambda\cdot \frac{1}{\ell_j}}\).

Suppose \(\Delta'\in (\frac{n^8\lambda}{2},n^8\lambda]\).
Recall \(\event'\) requires some shift of \(T\) such that each point of \(P'\) lies distance
\(\frac{\Delta'}{n^4}\) or greater from each side of \(C'\).
Therefore, \(\ell_j\ge \frac{2\Delta'}{n^4}>n^4\lambda\) for each \(j\).
Further,
\begin{align*}
\probability\Brack{\Delta_{C(p,q)} \in [\Delta,2\Delta) \mid \event'}
&\le \sum^d_{j=1}{\ceil{\frac{2\ell_j}{\Delta''}}\cdot \lambda\cdot
\frac{1}{\ell_j}}\\
&\le \sum^d_{j=1}{\frac{2\ell_j+\Delta''}{\Delta''}\cdot \lambda\cdot
\frac{1}{\ell_j}}\\
&< O\Paren{\frac{\lambda}{\Delta''}+\frac{1}{n^4}}.
\end{align*}

Because \(\Delta \in [\frac{\lambda}{\sqrt{d}},n^4\lambda]\), we see \(\Delta'' < 2n^4\lambda\),
implying \(\frac{\lambda}{\Delta''} > \frac{2}{n^4}\).
We have\linebreak
\(\probability\Brack{\Delta_{C(p,q)} \in [\Delta,2\Delta) \mid \event'} < O(\frac{\lambda}{\Delta})\).

If \(\Delta'\le \frac{n^8\lambda}{2}\), \(C'\) must be the root cell of the sub-quadtree \(T\).
In this case, all shifts of \(T\) will satisfy \(\event'\) assuming \(n\) is sufficiently large.
We have \(\ell_j = \frac{\Delta'}{3}\) and \(\Ceil{\frac{2\ell_j}{\Delta''}} \leq
\frac{3\ell_j}{2\Delta''}\) for all \(j\).
\begin{align*}
\probability\Brack{\Delta_{C(p,q)} \in [\Delta,2\Delta) \mid \event'}
&\le \sum^d_{j=1}{\ceil{\frac{2\ell_j}{\Delta''}}\cdot \lambda\cdot
\frac{1}{\ell_j}}\\
&= \sum_{j = 1}^d \frac{3\lambda}{2\Delta''}\\
&= O\Paren{\frac{\lambda}{\Delta}}.
\end{align*}
%
Either way, \(\probability\Brack{\Delta_{C(p,q)} \in [\Delta,2\Delta) \mid \event'} \leq
O(\frac{\lambda}{\Delta})\).

Again, \(\event\) immediately implies \(\event'\), but the converse is not necessarily true.
Therefore,
\begin{align*}
  \probability\Brack{\Delta_{C(p,q)} \in [\Delta,2\Delta) \mid \event}
  &= \frac{\probability[\Delta_{C(p,q)} \in [\Delta,2\Delta) \wedge \event]}{\probability[\event]}\\
  &\leq \frac{\probability[\Delta_{C(p,q)} \in [\Delta,2\Delta) \wedge \event']}{\probability[\event]}\\
  &= \frac{\probability[\Delta_{C(p,q)} \in [\Delta,2\Delta) \mid \event'] \cdot
  \probability[\event']}{\probability[\event]}\\
  &\leq \frac{\probability[\Delta_{C(p,q)} \in [\Delta,2\Delta) \mid \event']}{\probability[\event]}\\
  &\leq O\Paren{\frac{\lambda}{\Delta}} \cdot \frac{1}{1 - O(n / \log n)}\\
  &= O\Paren{\frac{\lambda}{\Delta}}.
\end{align*}

Assuming Property 2 of Lemma~\ref{lm:tree}, the value of \(\Delta_{C(p,q)}\) is in
\([\frac{\lambda}{\sqrt{d}}, n^4\lambda]\);
smaller cells cannot fit both points, and larger cells cannot separate them without one or both
points lying too close to a cell side.
We have
\begin{align*}
  \mathbb{E}(\Delta_{C(p,q)} \mid e) &\le\sum_{\Delta=2^i\cdot \frac{\lambda}{\sqrt{d}}, 0\le i\le
  \lg{(n^4\sqrt{d})}, i\in \mathbb{N}}\probability[\Delta_{C(p,q)}\in [\Delta,2\Delta) \mid e]\cdot 2\Delta \footnotemark\\
                        &<\sum_{\Delta=2^i\cdot \frac{\lambda}{\sqrt{d}}, 0\le i\le \lg{(n^4\sqrt{d})}, i\in \mathbb{N}} O(\frac{\lambda}{\Delta})\cdot 2\Delta\\
                        &\leq O(\log n) \cdot \lambda.
\end{align*}
\footnotetext{We use \(\Lg\) to denote the logarithm with base \(2\).}
We conclude
\begin{align*}
    \mathbb{E}(\distG(p,q) \mid e)&=||p-q||_2+\mathbb{E}(\Phi_{p,q} \mid e)\\
    &\le ||p-q||_2 + 3\sqrt{d}\eps_0\mathbb{E}(\Delta_{C(p,q)} \mid e)\\
    &\le (1+O(\eps_0\log n)) \cdot ||p-q||_2.
\end{align*}
\end{proof}


\subsection{Reduction to minimum cost flow}
\label{sec:spanner-flow}
Having built our sparse graph, we now reduce to a minimum cost flow problem in \(G^*\).
We model the minimum cost flow problem as follows to simplify later discussions.

Let \(G = (V, E)\) be an arbitrary undirected graph with \(V \in \R^d\).
Let \(\dartsof{E}\) be the set of edges in \(E\) oriented arbitrarily.
We call \(f \in \R^{\dartsof{E}}\) a \EMPH{flow vector} or more simply, a \EMPH{flow}.
Let \(A\) be a \(|V|\times |\dartsof{E}|\) \EMPH{vertex-edge incidence matrix} where \(\forall (u,(v,w))\in V\times\dartsof{E}\), \(A_{u,(v,w)}=1\) if \(u=v\), \(A_{u,(v,w)}=-1\) if \(u=w\), and \(A_{u,(v,w)}=0\) otherwise.
Given \(f\), we define the \EMPH{divergence} of a vertex \(v\) as \((A f)_{v} = \sum_{(v, w)} f_{(v,w)} - \sum_{(u,v)} f_{(u,v)}\).
For simplicity of exposition, we may sometimes refer to \(f_{(v, u)}\) even though \((u,v) \in \dartsof{E}\).
In such cases, it is assumed \(f_{(v,u)} = -f_{(u,v)}\).

Let \(||\cdot||_{\dartsof{E}}\) be a norm on \(\R^{\dartsof{E}}\) such that \(||f||_{\dartsof{E}}=\sum_{(u,v)\in \dartsof{E}}{|f_{(u,v)}| \cdot ||v - u||_2}\).
Let \(b \in \R^V\)  denote a set of divergences for all \(v\in V\).
We define an instance of \EMPH{uncapacitated minimum cost flow} as the pair \((G, b)\).
We seek a flow vector~\(f\) minimizing \(||f||_{\dartsof{E}}\) subject to \(Af=b\).


In particular, set \(b^* \in \R^V\) such that \(b^*_p=\mu(p), \forall p\in P\) and \(b^*_v=0, \forall v \in V\setminus P\).
Ultimately, we will find an approximate solution to the instance \((G^*, b^*)\).
Let \(\cost(G^*, b^*):=||f^*||_{\dartsof{E}}\) for some optimal solution \(f^*\) of this instance.
From construction of \(G^*\), \(\cost(P,\mu)\le\cost(G^*,b^*)\).
With high probability, the conditions of Lemma~\ref{lm:tree} hold true, and by Lemma~\ref{lm:espd},
\(\mathbb{E}(\cost(G^*,b^*))\le(1+O(\eps_0\log{n}))\cost(P,\mu)\).
In particular, \(\mathbb{E}(\cost(G^*, b^*) - \cost(P, \mu)) \leq O(\eps_0\log{n}) \cost(P, \mu)\).
We can guarantee that expected bound holds with high probability as well by doubling the constant in
the big-Oh and taking the best result from \(O(\log n)\) runs of our algorithm.
From here on, we assume both that the conditions Lemma~\ref{lm:tree} hold and
\(\cost(G^*,b^*)\le(1+O(\eps_0\log{n}))\cost(P,\mu)\).

\subsection{Decomposition into simpler subproblems}
\label{sec:spanner-decomposition}

In the sequel, we apply Sherman's generalized preconditioning framework~\cite{DBLP:conf/soda/Sherman17,khesin2021preconditioning} to find an approximate solution to the minimum cost flow instance \((G^*, b^*)\).
For technical reasons, however, we cannot afford to run the framework on the entire sparse graph \(G^*\) at once.

Here, we reduce finding an approximately optimal flow for minimum cost flow instance \((G^*, b^*)\) to finding \(O(n)\) approximately optimal flows, each within an induced subgraph defined by the net points within a simple sub-quadtree.

Recall, for each point \(p \in P\), \(\Tilde{C}(p)\) denotes the smallest subcell containing \(p\), and \(N_{\Tilde{C}}\) denotes the net point of subcell \(\Tilde{C}\).
Let \(f\) be the flow such that \(f_{(p, N_{\tilde{C}(p)})} = b^*_p\) for all \(p \in P\).
Let \(G' = (V', E')\) and \(A'\) be the restriction of \(G^*\) and its vertex-edge incidence matrix \(A\) after removing all vertices \(p \in P\).
Let \(b'\) be the restriction of \(b^* - Af\) to vertices of \(G'\).
Every vertex \(p \in P\) of \(G^*\) has exactly one incident edge, so an optimal solution to our original minimum cost flow instance consists of \(f\) along with an optimal solution to the instance defined on \(A'\) and \(b'\).
From here on, we focus on finding an approximately minimum cost flow in \(G'\).

Suppose there are multiple simple sub-quadtrees. Let \(G_0 = (V_0, E_0)\) be the subgraph induced by the \(m\) net point vertices of a simple sub-quadtree with no descendent sub-quadtrees.
Let \(C\) be the root cell of the simple sub-quadtree for \(G_0\), let \(u\) be a net point for an arbitrary subcell of \(C\), and let \(v\) be the parent net point of \(u\) in \(G'\) where \(C'\) is the subcell with \(v\) as its net point.
In \(O(m)\) time, we compute \(B = \sum_{w \in V_0} b'_w\), the total divergence of vertices within \(G_0\).
We then let \(f'\) be the flow in \(G'\) that is \(0\) everywhere except for \(f_{(u, v)} := B\).
Finally, let \(b'' = b' - A'f'\).

Notice that at least \(B\) units of flow in \(G_0\) need to leave or enter \(C\) by edges of side
length at least \(\Delta_{\Tilde{C}}\).
Given \(\Delta_{C}\le O(1/n^8)\Delta_{\Tilde{C}}\), we can lazily
assign the flow between net points of \(C\) and \(v\), increasing the cost by at most \(2\sqrt{d}\Delta_{C}B\le O(1/n^8)\Delta_{\Tilde{C}}B\).
We have the following lemma.
\begin{lemma}
  There exists a flow \(f''\) in \(G'\) such that \(f''_{(w,x)} = 0\) for all \(w \in V_0, x \notin V_0\); \(Af'' = b''\); and \(||f'' + f'||_{\dartsof{E'}} \leq (1 + O(1/n^8)) \cdot \cost(G', b')\).
\end{lemma}
\begin{proof}
Let \(\tilde{C}\) be the subcell for which \(v\) is a net point.
Let \(\Delta_{\Tilde{C}}\) be the side length of \(\Tilde{C}\).
By construction of \(G'\), at least \(B\) units of flow must travel to or from vertex \(v\) from \(G_0\) at a cost of \(\Delta_{\tilde{C}}\).
Specifically, \(G_0\) is totally inside \(\Tilde{C}\), \(v\) is the only vertex in \(\Tilde{C}\) incident to some edge crossing the side of \(\Tilde{C}\), and the nearest vertex \(x\notin V_0\) is at least \(\Delta_{\Tilde{C}}\) far from \(v\). So \(\cost(G',b')\ge \Delta_{\Tilde{C}}B\).

Suppose \(f^{'*}\) is a flow in \(G'\) with cost \(\cost{(G',b')}\).
Let \(N_C\) be the set of net points of subcells of \(C\).
We may assume there is no pair \(y, z \in N_C\) such that \(f_{(y,v)} > 0\) and \(f_{(v,z)} > 0\), because we could send the flow directly between \(y\) and \(z\) more cheaply.
We create flow \(f'''\) as follows starting with \(f'' = f^{'*}\).
While there exists some vertex \( u'\in N_C\backslash\{u\}\) with \(f'''_{(u', v)}\ne 0\), let \(\delta = f'''_{(u', v)}\).
We divert flow by setting \(f'''_{(u,v)} \gets f'''_{(u,v)} + \delta\),  \(f'''_{(u', u)} \gets f'''_{(u', u)} + \delta\), and \(f'''_{(u', v)} \gets 0\).
This increases the cost by at most twice of the length of the diagonal of \(C\) per diverted unit of flow.
Overall, we divert at most \(B\) units.
The total cost increase is at most \(2\sqrt{d}\Delta_CB \le O(1/n^8)\cost{(G',b')}\) where \(\Delta_C\) is side length of \(C\), because \(\Delta_{C}\le O(1/n^8)\Delta_{\Tilde{C}}\).
We have \(||f'''||_{\dartsof{E}}\le(1+O(1/n^8))\cdot\cost(G',b')\).
Finally, let \(f''=f'''-f'\).
\end{proof}

The above lemma implies we can use the following strategy for approximating a minimum cost flow in \(G'\):
Let \(b_0\) be the restriction of \(b''\) to \(V_0\).
We find a flow in \(G_0\) with divergences \(b_0\) of cost at most \((1 + O(\eps)) \cdot \cost(G_0, b_0)\) using the algorithm described in the next section.
Then, we recursively apply our algorithm on \(G'' = (V'', E'')\), the induced subgraph over \(V'' = V' \setminus V_0\).
The depth of recursion is \(O(n)\), so the total cost from combining our separately computed flows is \((1 + O(\eps))(1 + O(1/n^7)) \cdot \cost(G', b') = (1 + O(\eps)) \cost(G', b')\).
We must emphasize that simple sub-quadtrees may still have linear depth, so we still need to apply our own techniques to make Sherman's framework run within the desired time bounds.


\section{Approximating the minimum cost flow}\label{sec:sol}
Let \(G = (V, E)\) be an induced subgraph of sparse graph \(G^*\) where \(V\) is the subset of net points for one simple sub-quadtree \(T\) as defined above.
Let \(m = |E|\), and let \(A\) be the vertex-edge incidence matrix for \(G\).
We now describe the ingredients we need to provide to efficiently approximate the minimum cost flow problem in \(G\) using Sherman’s generalized preconditioning framework \cite{khesin2021preconditioning, DBLP:conf/soda/Sherman17}.
We then provide those ingredients one-by-one to achieve a near-linear time \((1+O(\eps))\)-approximate solution for the minimum cost flow instance. 

\subsection{The preconditioning framework}

Consider an instance of the minimum cost flow problem in \(G\) with an arbitrary divergence vector \(\tilde{b} \in \R^{V}\), and let \(f^*_{\Tilde{b}}:=\Argmin_{f\in\R^{\dartsof{E}}, Af=\Tilde{b}}{||f||_{\dartsof{E}}}\).
A flow vector \(f\in\R^{\dartsof{E}}\) is an \EMPH{\((\alpha, \beta)\) solution} to the problem if
\begin{align}
    \nonumber ||f||_{\dartsof{E}}&\le \alpha||f^*_{\Tilde{b}}||_{\dartsof{E}}\\
    \nonumber ||Af-\tilde{b}||_1&\le \beta||A||\,||f^*_{\Tilde{b}}||_{\dartsof{E}}
\end{align}
where \(||A||\) is the norm of the linear map represented by \(A\).
An algorithm yielding an \((\alpha, \beta)\)-solution is called an \EMPH{\((\alpha, \beta)\)-solver}.

By arguments in \cite{khesin2021preconditioning}, we seek a preconditioner \(B\in \R^{V\times V}\) of full column rank such that, for any \(\Tilde{b}\in\R^V\) with \(\sum_{v\in V}{\Tilde{b_v}}=0\), it satisfies

\begin{equation}\label{eq:3.3}
||B\Tilde{b}||_1\le \Min\{||f||_{\dartsof{E}}:f\in \R^{\dartsof{E}}, Af=\Tilde{b}\}\le \kappa ||B\Tilde{b}||_1
\end{equation}
for some sufficiently small function \(\kappa\) of \(n\), \(\eps\), and \(d\).

Let \(M\) be the time it takes to multiply \(BA\) and \((BA)^T\) by a vector. Then there exists a \((1+\eps, \beta)\)-solver for any \(\eps, \beta>0\) for this problem with running time bounded by \(O(\kappa^2(|V|+|\dartsof{E}|+M)\log{|\dartsof{E}|}(\eps^{-2}+\log{\beta^{-1}})\) \cite{DBLP:conf/soda/Sherman17}. Moreover, if a feasible flow \(f\in \R^{\dartsof{E}}\) with cost \(||f||_{\dartsof{E}}\le \kappa B\Tilde{b}\) can be found in time \(K\), there is a \((\kappa, 0)\)-solver with running time \(K\).
By setting \(\beta=\eps\kappa^{-2}\) \cite{khesin2021preconditioning}, the composition of these two solvers is a \((1+2\eps, 0)\)-solver with running time bounded by

\begin{displaymath}
O(\kappa^2(|V|+|\dartsof{E}|+M)\log{|\dartsof{E}|}(\eps^{-2}+\log{\kappa})+K).
\end{displaymath}

\subsection{Preconditioning the minimum cost flow}

We present a way to construct such a preconditioner \(B\) similar to the one of Khesin \emph{et al.}~\cite{khesin2021preconditioning} that guarantees \(\kappa\) in \eqref{eq:3.3} is sufficiently small for our performance objective. Our algorithm \textit{does not} compute \(B\) directly, because \(B\) is not sparse. However, the time for individual applications of \(BA\) or \((BA)^T\) is \(O(|V|+|\dartsof{E}|)\).

Let \(\Tilde{\mathbb{C}}\) denote the set of all subcells defining the net points of \(G\).
For any subcell \(\tilde{C}\in \tilde{\mathbb{C}}\), let \(N_{\tilde{C}}\) denote its net point and let \(\Delta_{\tilde{C}}\) denote its side length.

Let \(B\) be a matrix indexed by \((u, v) \in V \times V\) such that,
for every net point \(\nu\) in \(V\) where \(\nu\) is the net point of some subcell \(\Tilde{C}\), we set \(B_{\nu, v}=\frac{\Delta_{\Tilde{C}}}{\Lambda}\) for all descendent net points \(v\) of~\(\nu\), where \(\Lambda=22 \lg (\frac{n}{\eps_0})\).
\(B_{\nu, v} = 0\) for all other \(v\).
Matrix \(B\) has full column rank, because each column specifies exactly which ancestor net points each vertex has in \(G\).

Now, fix any \(\Tilde{b}\in \R^V\) such that \(\sum_{v\in V}{\Tilde{b}_v}=0\).
Observe,
\begin{equation}
\label{eq:cond} ||B\Tilde{b}||_1=
\sum_{\Tilde{C}\in \Tilde{\mathbb{C}}}{\frac{\Delta_{\Tilde{C}}}{\Lambda}{|\sum_{v\in \Tilde{C}}{\Tilde{b}_v}}|}.
\end{equation}

\begin{lemma}\label{lm:below}
We have \(||B\Tilde{b}||_1\le \Min\{||f||_{\dartsof{E}}:f\in \R^{\dartsof{E}}, Af=\Tilde{b}\}\).
\end{lemma}

\begin{proof}
Let \(f^*_{\Tilde{b}}:=\Argmin_{f\in\R^{\dartsof{E}}, Af=\Tilde{b}}{||f||_{\dartsof{E}}}\).
%
We arbitrarily decompose \(f^*_{\tilde{b}}\) into a set of flows \(F = \Set{f^1, f^2, \dots}\) with the following properties: 1) each flow follows a simple path between two vertices \(u\) and \(v\); 2) for each flow \(f^i \in F\) and edge \((u, v) \in \dartsof{E}\) either \(f^i(u, v) = 0\) or its sign is equal to the sign of \(f^*_{\tilde{b}}(u,v)\); 3) for each flow \(f^i \in F\) and vertex \(v\), either \((Af^i)_v = 0\) or its sign is equal to \(\tilde{b}_v\); and 4) for each edge \((u,v) \in \dartsof{E}\), we have \(f^*_{\tilde{b}}(u, v) = \sum_{f^i \in F} f^i(u, v)\).
The existence of such a decomposition is a standard part of network flow theory and one can be computed in a simple greedy manner (however, our algorithm does not actually need to compute one).
From construction, we have \(\sum_{f^i \in F} ||f^i||_{\dartsof{E}} = ||f^*_{\tilde{b}}||_{\dartsof{E}}\).
We describe a way to charge summands of \(\sum_{\Tilde{C}\in \Tilde{\mathbb{C}}} \Delta_{\Tilde{C}} |\sum_{v\in \Tilde{C}} \Tilde{b}_v|\) to the summands of \(\sum_{f^i \in F} ||f^i||_{\dartsof{E}}\).
Our charges will cover each of the former and exceed each of the latter by at most a \(\Lambda\) factor.
Consider a subcell \(\tilde{C}\).
For each vertex \(u \in \tilde{C}\), for each flow \(f^i\) sending flow to or from \(u\), we charge \(\Delta_{\Tilde{C}} |(Af^i)_u|\).
Clearly, we charge at least \(\Delta_{\Tilde{C}} |\sum_{v\in \Tilde{C}}{\Tilde{b}_v}|\) for each subcell \(\tilde{C}\).




It remains to prove we did not overcharge by too large a factor.
Consider an arbitrary flow \(f^i \in F\) sending flow from some vertex \(u\) to some vertex \(v\).
Let \(C(u,v)\) be the lowest common ancestor cell containing \(u\) and \(v\).
Let \(\Delta_{C(u,v)}\) be its side length, and let \(C(\hat{u}, v)\) be the child cell of \(C(u,v)\) that includes \(u\).
Let \(\Delta\) be the side length of \(C(\hat{u}, v)\).

Suppose there exists a descendant cell \(C'\) of \(C(\hat{u}, v)\) containing \(u\) that is at least \(5\Lg{n}\) levels down from \(C(\hat{u}, v)\).
Its side length \(\Delta_{C'}\) is at most \(\frac{\Delta}{n^5}\).
By Property 2 of Lemma~\ref{lm:tree}, \(u\) is at least \(\frac{\Delta}{n^4}\) 
distance away from any side of \(C(\hat{u}, v)\) and therefore \(v\) as well.
Therefore, we charge at most an \(\frac{\eps_0}{n}\) fraction of \(||f^i||_{\dartsof{E}}\) to cover \(u\)'s subcell in \(C'\).
The amounts charged by similar subcells of smaller side length containing \(u\) form a decreasing
geometric series evaluating to at most that value, so all these small subcells charge at most a
\(\frac{2\eps_0}{n}\) fraction total.


Now, consider the cells with larger side length.
Suppose there exists an ancestor cell \(C''\) of \(C(\hat{u}, v)\) at least \(\Lg{\eps_0^{-1}} + 1\) levels up from \(C(\hat{u}, v)\), and let \(\Tilde{C}''\) be the subcell of \(C''\) containing \(u\).
Then the side length of \(\Tilde{C}''\) is at least \(\Delta_{C(u, v)}\) and all points in \(C(u, v)\) will be included in \(\Tilde{C}''\) also.
Therefore, we do not charge to \(||f^i||_{\dartsof{E}}\) for subcell \(\tilde{C}''\), and there are at most \(5\Lg{n}+\Lg{\eps_0^{-1}} \leq 5\Lg{\frac{n}{\eps_0}}\) subcells in addition to those handled above for which we do charge to \(||f^i||_{\dartsof{E}}\).
Consider any such subcell \(\tilde{C}\).
The path carrying \(f^i\) leaves \(\tilde{C}\) through an edge of length at least \(\Delta_{\tilde{C}} / 2\), so we charge at most \(2 \cdot ||f^i||_{\dartsof{E}}\) to cover \(\tilde{C}\).
Summing over all \(5\Lg{\frac{n}{\eps_0}}\) choices of \(\tilde{C}\) and accounting for the tiny
cells as discussed above, we charge at most \((10\Lg{\frac{n}{\eps_0}} + 2\eps_0 / n) ||f^i||_{\dartsof{E}} \leq 11\Lg{(\frac{n}{\eps_0})} \cdot ||f^i||_{\dartsof{E}}\) to cover subcells containing \(u\).
We also charge to \(||f^i||_{\dartsof{E}}\) to cover subcells containing \(v\), so we overcharge by a factor of at most \(22\Lg{(\frac{n}{\eps_0})}  = \Lambda\).
The lemma follows.
\end{proof}




\begin{lemma}\label{lm:above}
We have \(\Min\{||f||_{\dartsof{E}}:f\in \R^{\dartsof{E}}, Af=\Tilde{b}\}\le \kappa ||B\Tilde{b}||_1\) for some\linebreak
\(\kappa=O(\eps_0^{-1} \log{(n/\eps_0)})\).
Moreover, a flow vector \(f\) satisfying \(Af = \tilde{b}\) of cost at most \(\kappa ||B\Tilde{b}||_1\) can be computed in 
\(O(m)\) time.
\end{lemma}

\begin{proof}
We describe a greedy algorithm based on one by Khesin \emph{et al.}~\cite{khesin2021preconditioning} to iteratively construct a feasible flow \(f\) satisfying \(Af=\Tilde{b}\) with a cost \(||f||_{\dartsof{E}}\le \kappa B\Tilde{b}\) in \(O(m)\) time.
At any point during \(f\)'s construction, we say the \EMPH{surplus} of vertex \(u \in V\) is \(\pi(u, f)=(Af)_u - \Tilde{b}_{u}\), the difference between the current and desired divergences of \(u\).

\begin{enumerate}
    \item For every cell \(C\) in a postorder traversal of \(G\)'s simple sub-quadtree, for every subcell \(\Tilde{C}\) of \(C\), we do the following. Let \(\nu=N_{\Tilde{C}}\). We choose any two child net points \(v, w\) of \(\nu\) such that \(\pi(v, f)>0>\pi(w, f)\). We then add \(\Min\{|\pi(v, f)|, |\pi(w, f)|\}\) to \(f_{(w,v)}\).
In doing so, we make the surplus of at least one child net point of \(\nu\) equal to \(0\), and we decrease the absolute values of surpluses of both \(v\) and \(w\).
    Therefore, after at most a number of steps equal to the number of child net points of \(\nu\), either all child net points have non-negative surplus or all child net points have non-positive surplus.
    Finally, for each vertex \(v\) among child net points with non-zero surplus, we set \(f_{(\nu , v)}=\pi(v, f)\).
    Afterward, every child net point of \(\nu\) has surplus \(0\).
    In other words, the unbalance among those child net points is collected into \(\nu\).
    Each net point \(\nu\) has at most \(2^d\) child net points.
    Therefore, the total running time for this step is \(O(m)\).
    
    \item After performing step 1), all net points with parents have a surplus of \(0\).
    We pick up any two net points \(u\), \(v\) of subcells of \(T\)’s root cell with two different surplus signs as described in step \(2\) and add \(\Min\{|\pi(u, f)|, |\pi(v, f)|\}\) to \(f_{(v,u)}\).
    After \(O(\eps_0^{-d}) = O(m)\) steps, all points \(v\in V\) will have surplus \(0\), and \(f\) is a feasible flow satisfying \(Af=\Tilde{b}\).
\end{enumerate}

We now analyze \(||f||_{\dartsof{E}}\).
Consider a subcell \(\tilde{C}\) of some cell \(C\) with net point \(\nu\).
Flow does not leave or enter \(\tilde{C}\) until we move flow between \(\nu\) and either another net point in \(C\) or \(\nu\)'s parent net point.
Therefore, \(\pi(\nu, f) = -\sum_{v\in \Tilde{C}}{\Tilde{b}_v}\) immediately after moving flow from \(\nu\)'s children to \(\nu\) in step 1) above.
All subsequent steps moving flow to or from \(\nu\) involve an edge of length at most \(\eps_0^{-1}\sqrt{d}\Delta_{\tilde{C}}\) and only serve to reduce \(|\pi(\nu, f)|\).

Summing over all subcells, we get

\begin{equation}
\nonumber ||f||_{\dartsof{E}}\le
\sum_{\Tilde{C}\in \Tilde{\mathbb{C}}} \eps_0^{-1}\sqrt{d}\Delta_{\tilde{C}} |\sum_{v\in \Tilde{C}}{\Tilde{b}_v}| \leq \eps_0^{-1}\sqrt{d}\Lambda||B\Tilde{b}||_1.
\end{equation}

Therefore, \(||f^*_{\Tilde{b}}||_{\dartsof{E}}\le||f||_{\dartsof{E}}\le\kappa ||B\Tilde{b}||_1\), where \(\kappa=O(\eps_0^{-1}\log{(n/\eps_0)})\).
\end{proof}

\begin{lemma}\label{lm:cond}
Applications of \(BA\) and \((BA)^T\) to arbitrary vectors \(f \in \R^{\dartsof{E}}\) and \(\tilde{b} \in \R^V\), respectively, can be done in 
\(O(m)\) time.
\end{lemma}

\begin{proof}
Both applications can be performed using dynamic programming algorithms.

\paragraph*{Computing \(BAf\)} Let \(A'=Af\). Recall, \(\forall v\in V\), \(A'_v\) is the divergence of \(v\) given flow \(f\).
Matrix \(A\) has \(m\) non-zero entries, so \(A'\) can be computed in \(O(m)\) time.

We compute \(BAf\) by computing \(BA'\).
Let \(\nu\) be any net point of \(G\), and let \(\tilde{C}\) be its subcell.
From the definition of \(B\), we have \((BA')_{\nu} = \frac{\Delta_{\tilde{C}}}{\Lambda} \sum_{v \in \tilde{C}} A'_v\).
Now, let \(\tilde{C}^+\) be the (possibly empty) set of all child subcells of \(\tilde{C}\) with net points in \(G\).
We have \(\sum_{v\in \Tilde{C}}{A'_v} = A'_{\nu} + \sum_{\Tilde{C}'\in \Tilde{C}^+}{\sum_{v\in\Tilde{C}'} {A'_v}}\).
Thus, we can use dynamic programming to compute \(BA'\) in \(O(m)\) time.
Each entry is filled in during a postorder traversal of the quadtree cells.

\paragraph*{Computing \((BA)^T \tilde{b}\)} 
Recall, \((BA)^T = A^T B^T\).
Let \(b'=B^T\tilde{b}\).
We begin by computing \(b'\).
Let \(\tilde{C}\) be any subcell with a net point in \(G\), and let \(\nu = N_{\tilde{C}}\).
Let \(\Tilde{C}^-\) be the set of all ancestor subcells of \(\Tilde{C}\) with net points in \(G\) including~\(\Tilde{C}\).
We have \(b'_{\nu} = \sum_{\tilde{C}' \in \tilde{C}^-} \frac{\Delta_{\tilde{C}'}}{\Lambda} \tilde{b}_{N_{\tilde{C}'}} = \frac{\Delta_{\tilde{C}}}{\Lambda} \tilde{b}_{\nu} + b'_{N_{\Tilde{C}^{p}}}\).
Therefore, we can use dynamic programming to compute \(b'\) in \(O(m)\) time.
Each entry is filled in during a \emph{pre}order traversal of the quadtree cells.
Finally, \(A^T\) has \(m\) non-zero entries, so \(A^T B^T \tilde{b} = A^T b'\) can be computed in \(O(m)\) time as well.
\end{proof}

We have shown there exists a \((1+2\eps, 0)\)-solver for the minimum cost flow problem on \(G\).
Plugging in all the pieces, we get a running time bounded by

\[
O(m\eps_0^{-2}\log^3{(n/\eps_0)(\eps^{-2}+\log{(n/\eps_0)}))}.
\]

Recall, \(\eps_0 = O(\eps/\log{n})\).
We run the preconditioning framework algorithm in each graph \(G\) induced by a simple sub-quadtree's net points as described in Section~\ref{sec:spanner-decomposition}.
The final running time to compute a flow in \(G^*\) of cost at most \((1 + \eps) \cost(P, \mu)\) is
\[
n\eps^{-O(d)}\log^{O(d)}{n}.
\]


\section{Recovering a transportation map from the minimum cost flow}\label{sec:recover}
We now describe how to recover a transportation map of \(P\) using the approximately minimum cost flow \(\hat{f} \in \dartsof{E}\) we computed for \(G^*\).
Unlike~\(\hat{f}\), the transportation map \(\tau\) contains only weighted pairs of points in \(P\).
We will \emph{implicitly} maintain a flow \(f\) of cost at most \(||\hat{f}||_{\dartsof{E^*}}\) that will eventually describe our transportation map.
Abusing notation, we extend the definition of \(f_{(u,v)}\) to include any pair of vertices in \(G^*\).
Value \(f_{(u,v)}\) is initially \(0\) for all \(uv \notin E^*\).
We follow the strategy of Khesin \etal~\cite{khesin2021preconditioning} of iteratively rerouting flow going through each net point~\(\nu\) to instead go directly between vertices receiving from or sending flow to~\(\nu\), eventually resulting in no flow going through any net point.
Nearly every pair containing a point \(p \in P\) and an ancestor net point may at some moment carry flow during this procedure.
Because quadtree \(T^*\) has such high depth, we must take additional care.

To quickly maintain these flow assignments with points in \(P\), we store two data structures \(pt(\nu)\) and \(nt(\nu)\) for each net point \(\nu\in V^* \setminus P\).
We call these data structures the \EMPH{prefix split trees} of \(\nu\).
The prefix split tree is stored as an ordered binary tree data structure where each node has a weight.
We let \(w(x)\) denote the weight of node \(x\) in a tree~\(S\) and \(w(S)\) denote the total weight of all nodes in \(S\).
These trees support the standard operations of insertion and deletion.
They support changing the weight of a single node.
They support the \(\textsc{Merge}(S, S')\) operation which takes two trees \(S\) and \(S'\) and combines them into one tree with all members of \(S\) appearing in order before \(S'\).
Finally, they support the \(\textsc{PrefixSplit}(S, t)\) operation defined as follows.
Given a target value \(t\) and a prefix split tree \(S\), \textsc{PrefixSplit} finds a maximal prefix of \(S\)'s nodes in order where the sum of node weights in the subset is less than or equal to \(t\).
If the sum is less than \(t\), it splits the next node~\(x\) into two nodes~\(x_1\) and~\(x_2\) where \(w(x_1) + w(x_2) = w(x)\).
The split guarantees adding \(x_1\) to the maximal prefix subset makes the sum weight of the subset exactly equal to \(t\).
The operation then splits off all members of this subset, including \(x_1\) if a node \(x\) was split, into their own tree \(S'\) and returns it, leaving \(S\) with only the remaining nodes.
We emphasize that the order of nodes within the data structure is important for defining \textsc{PrefixSplit}, but the nodes are not ``sorted'' in any meaningful sense;
in particular, any two trees can be merged as defined above.
All those operations can be done in amortized \(O(\log m)\) time, where \(m\) is the number of nodes in the tree, by applying simple modifications to the splay tree data structure of Sleator and Tarjan~\cite{DBLP:journals/jacm/SleatorT85}.
We provide details on how to implement a prefix split tree in Appendix~\ref{app:prefix}.

In our setting, every node in \(pt(\nu)\) and \(nt(\nu)\) represents a point \(p \in P\).
Thanks to our use of the \textsc{PrefixSplit} procedure, some points may be represented multiple times in a single tree.
We use \(pt(\nu)[p]\) to denote the set nodes representing \(p\) in \(pt(\nu)\), and define \(nt(\nu)[p]\) similarly.
Our algorithm implicitly maintains the invariant that for all net points \(\nu\) and points \(p \in P\), \(\sum_{x\in pt(\nu)[p]}{w(x)}-\sum_{x\in nt(\nu)[p]}{w(x)}=f_{(\nu, p)}\).
We proceed with the algorithm given in Figure~\ref{fig:recover}.

\begin{figure}[t]
\begin{algo}
    \Comment{Initialize data structures.}
    \\ For all net points \(v \in V^* \setminus P\) and \(p \in P\) where \(f_{(p, v)} > 0\)\+
    \\    Insert a node of weight \(f_{(p,v)}\) into \(nt(v)\) representing \(p\)\-
    \\  For all net points \(v \in V^* \setminus P\) and \(p \in P\) where \(f_{(v, p)} > 0\)\+
    \\  Insert a node of weight \(f_{(v,p)}\) into \(pt(v)\) representing \(p\)\-
    \\  Let \(\mathbb{C}\) be the set of all cells
    \\  For \(C \in \mathbb{C}\) in postorder \+
    \\    Let \(N_C=\{N_{\Tilde{C}}: \Tilde{C}\in C\}\), \(N'_{C}=\{\text{parent of }v: v \in N_C\}\)
    \\  For each \(v \in N_C\)\+
    \\  \Comment{Cancel flow to/from other net points.}
    \\    While \(\exists u,w\in N_C\cup N'_{C}: f_{(v,w)}>0>f_{(v, u)}\) \+
    \\      \(\delta\gets \Min\{f_{(u, v)}, f_{(v, w)}\}\)
    \\      \(f_{(u, w)}\gets f_{(u, w)}+ \delta\)
    \\      \(f_{(u, v)}\gets f_{(u, v)} - \delta\)
    \\      \(f_{(v, w)}\gets f_{(v, w)} - \delta\)\-
    \\  \Comment{Now, either all other net points send flow to \(v\) or all get flow from \(v\).}
    \\  While \(\exists u \in N_C\cup N'_{C}: f_{(u, v)}>0\)\+
    \\  \Comment{Implicitly reduce \(f(v, p)\) and increase \(f(u, p)\) for several \(p \in P\)}
    \\    \(pt'\gets \textsc{PrefixSplit}(pt(v), f_{(u, v)})\)
    \\    \(\textsc{Merge}(pt', pt(u))\)\-
    \\  While \(\exists w \in N_C\cup N'_{C}: f_{(v, w)}>0\)\+
    \\  \Comment{Implicitly reduce \(f(p, v)\) and increase \(f(p, w)\) for several \(p \in P\)}
    \\    \(nt'\gets \textsc{PrefixSplit}(nt(v), f_{(v, w)})\)
    \\    \(\textsc{Merge}(nt', nt(w))\)\-
    \\  \Comment{Now, all flow to/from \(v\) involves points \(p \in P\).}
    \\  While \(pt(v)\) and \(nt(v)\) are not empty\+
    \\    Let \(x \in nt(v)[p], y \in pt(v)[q]\) for some \(p, q \in P\)
    \\    \(\delta \gets \Min\{w(x), w(y)\}\)
    \\    \(f_{(p, q)}\gets f_{(p, q)}+ \delta\)
    \\    \(w(x) \gets w(x) - \delta\); if \(w(x) = 0\), delete \(x\) from \(nt(v)\)
    \\    \(w(y) \gets w(y) - \delta\); if \(w(y) = 0\), delete \(y\) from \(pt(v)\)\-\-\-
    \\  For all \((p, q) \in P \times P\) where \(f_{(p, q)} > 0\)\+
    \\    \(\tau(p, q) \gets f_{(p,q)}\)
\end{algo}
\caption{Recovering a transportation map from an approximately minimum cost flow in \(G^*\).}
\label{fig:recover}
\end{figure}

\begin{lemma}
\label{lm:recover}
The algorithm described above produces a transportation map with cost at most \(||\hat{f}||_{\dartsof{E^*}}\) in \(n\eps_0^{-2d}\log^2 (n/\eps_0)\) time.
\end{lemma}

\begin{proof}

As stated, our algorithm implicitly maintains a flow \(f\) such that for all net points \(\nu\) and points \(p \in P\), \(\sum_{x\in pt(\nu)[p]}{w(x)}-\sum_{x\in nt(\nu)[p]}{w(x)}=f_{(\nu, p)}\).
After every iteration of any of the while loops, the divergences among all vertices in \(G^*\) remain the same.
Further, after processing any net point \(v\) in the inner for loop, there are no other vertices \(u\) in \(V^*\) such that \(f_{(u,v)}\neq 0\).
Observe the algorithm never changes the flow coming into or out of a net point \(v\) unless \(f_{(u,v)} \neq 0\) for some vertex \(u\).
Therefore, after \(v\) is processed, it \emph{never} has flow going into or out of it again (Khesin \etal~\cite{khesin2021preconditioning} refer to this property as \(v\) having \emph{uniform flow parity}).
Because we eventually process every net point in \(G^*\), we eventually end up with a flow \(f\) such that \(f_{(p,q)} \neq 0\) only if \(p,q \in P\).
We immediately see \(\tau\) is a transportation map.

To analyze the cost of \(\tau\), observe that after every iteration of a while loop, we replace some \(\delta\) units of flow passing through \(v\), possibly between multiple sources and one destination or vice versa, with \(\delta\) units going directly from the source(s) to the destination(s).
By the triangle inequality, this new way to route flow is cheaper, so the final flow \(f\), and subsequently \(\tau\) has smaller cost than \(\hat{f}\).

To implement our algorithm quickly, we only explicitly store new flow values whenever we have a line ``\(f_{(u, w)} \gets \_\)'' for some pair of vertices \((u,w)\).
Observe that every time we finish processing a cell, every one of its net points is also processed.
By the above discussion, flow no longer passes through those net points.
Therefore, as we process the net points for a cell \(C\), we never send flow from a net point \(v \in N_C\) to a net point outside \(N_C \cup N'_{C}\).
Every time we change flow going through another net point while processing a net point \(v\), we
decrease the number of net points \(u\) such that \(f_{(u,v)} \neq 0\) by one.
There are \(O(n \eps_0^{-d} \log (n/\eps_0))\) net points, and \(O(\eps_0^{-d})\) other net points in each \(N_C \cup N'_C\), so the number of iterations total in the first three while loops is \(O(n \eps_0^{-2d} \log (n/\eps_0))\).
Finally, observe that we only do \(\textsc{PrefixSplit}\) operations during these while loops, implying we create a total of \(O(n \eps_0^{-2d} \log (n/\eps_0))\) nodes throughout all prefix split trees.
Every iteration of the fourth while loop results in deleting a node from at least one of \(nt(v)\) or \(pt(v)\), so the number of iterations of this while loop is \(O(n \eps_0^{-2d} \log (n/\eps_0))\) as well.
Finally, every while loop iteration consists of a constant number of constant time operations in
addition to a constant number of prefix split tree operations, each of which can be done in \(O(\log
(n/\eps_0))\) amortized time.
\end{proof}

Recall the procedure of Section~\ref{sec:sol} will find a flow of cost at most \((1 + \eps) \cost(P, \mu)\) in \(n\eps^{-O(d)}\log^{O(d)}{n}\) time.
The procedure described above will then extract a transportation map from it. We conclude the proof of Theorem~\ref{theorem-result}.

\paragraph*{Acknowledgement} The authors would like to thank Hsien-Chih Chang for some helpful discussions that took place with the first author at Dagstuhl seminar 19181 ``Computational Geometry''.
We would also like to thank the anonymous reviewers of the conference version of this
paper~\cite{fl-ntasg-20} as well as this full version for many helpful comments and suggestions.


\bibliographystyle{abbrv}
\bibliography{references}

\appendix
\section{Prefix split trees}\label{app:prefix}
We implement our prefix split trees by modifying the splay tree data structure of Sleator and Tarjan~\cite{DBLP:journals/jacm/SleatorT85}.
Let \(S\) be a prefix split tree.
We store the weight~\(w(x)\) of each node \(x\) directly with the node itself.
Moreover, every node \(x\) keeps another value \(W(x)\) equal to the sum weight of all the descendants of \(x\) including \(x\) itself.

A \EMPH{splay} of a node \(x\) in \(S\) is a sequence of double rotations (possibly followed by a standard single rotation) that move \(x\) to the root of \(S\).
Only those nodes on the path from the root to \(x\) have their children pointers updated by a splay.
We can update \(W(y)\) for every such node \(y\) with only a constant factor overhead in the time to perform a splay.
Let \(s(x)\) denote the number of descendants of \(x\) in its prefix split tree, and let \(r(x) = \Floor{\lg s(x)}\).
Let \(\Phi(S) = \sum_{x \in S} r(x)\).
The \EMPH{amortized time} for an operation on \(S\) can be defined as the real time spent on the operation plus the net change to \(\Phi(S)\) after the operation.
The amortized time for a splay in an \(m\)-node tree is \(O(\log m)\)~\cite{DBLP:journals/jacm/SleatorT85}.

Recall, the order of nodes within a tree is largely irrelevant outside the definition of the \(\textsc{PrefixSplit}\) operation.
To insert a node \(x\) in \(S\), we add \(x\) as the child of an arbitrary leaf of \(S\) and splay \(x\) to the root.
The number of operations in the splay dominates, so the amortized cost of insertion is \(O(\log m)\).
To delete a node \(x\), we splay \(x\) to the root and delete it, resulting in two disconnected subtrees \(S_1\) and \(S_2\).
We then perform a \(\textsc{Merge}(S_1, S_2)\) in \(O(\log m)\) amortized time as described below, so the whole deletion has amortized cost \(O(\log m)\).
To update the weight of a node \(x\), we splay \(x\) to the root and update \(w(x)\) and \(W(x)\) in constant time each.
The splay once again dominates, so the total amortized cost is \(O(\log m)\).

The operation \(\textsc{Merge}(S_1, S_2)\) is implemented as follows.
Let \(x\) be the rightmost leaf of \(S_1\).
We splay \(x\) to the root so it has exactly one child.
We then make the root of \(S_2\) the other child of \(x\).
Let \(m\) be the total number of nodes in \(S_1\) and \(S_2\).
Adding the root of \(S_2\) as a child increases \(\Phi(S_1) + \Phi(S_2)\) by \(O(\log m)\), so the
amortized time for the \textsc{Merge} is \(O(\log m)\).

Finally, we discuss the implementation of \(\textsc{PrefixSplit}(S, t)\).
We assume \(t > 0\).
We use the values \(W(\cdot)\) to find the prefix of nodes desired.
Let \(y\) be the next node in order after the prefix.
We splay \(y\) to the root of \(S\).
Let \(x\) be the left child of \(y\) after the splay (if it exists).
Suppose \(W(x) < t\).
We delete \(y\), creating two trees \(S_1\) and \(S_2\) where \(S_1\) contains the nodes in the prefix.
We create a new node \(y_1\) of weight \(t - W(y)\) and make the root of \(S_1\) its child so that \(y_1\) is the new root.
We create a node \(y_2\) of weight \(w(y) - w(y_1)\) and make the root of \(S_2\) its child.
Now, suppose instead \(W(x) = t\).
In this case, we simply remove the edge between \(x\) and \(y\) to create a subtree \(S_1\) with \(x\) as its root.
Let \(S_2\) be the remainder of \(S\).
Whether or not \(W(x) = t\), we return \(S_1\) and set \(S = S_2\).
The amortized time for the \textsc{PrefixSplit} is the amortized time for a single splay and a constant number of edge changes, implying the \textsc{PrefixSplit} takes \(O(\log m)\) amortized time total.


\end{document}